\documentclass[a4paper,UKenglish,cleveref, autoref, unicode]{lipics-v2019}
\usepackage{algorithm}
\usepackage{algorithmicx}
\usepackage{algpseudocode}
\usepackage{algpascal}
\usepackage{booktabs}

\usepackage{amsmath}
\usepackage{amssymb}
\usepackage{color}
\definecolor{red}{RGB}{255,0,0}
\definecolor{blue}{RGB}{0,0,255}
\definecolor{green}{RGB}{0,255,0}

\newcommand {\abs}[1]  {\left\vert#1\right\vert}
\newcommand {\set}[1]  {\left\{#1\right\}}

\newcommand {\defined} {\stackrel{def} {=}}

\newcommand {\nph}     {\textsc{NP}\textrm{-hard}}
\newcommand {\apxh}    {\textsc{APX}\textrm{-hard}}
\newcommand {\maxsnph}    {\textsc{Max-Snp}\textrm{-hard}}

\newcommand {\ptas}     {\textsc{PTAS}}
\newcommand {\nohyp}   {\textsc{P}=\textsc{NP}}

\newcommand {\runningtitle}[1] {\vspace{0.5ex}\noindent{\textbf{\boldmath #1:}}}

\renewcommand {\comment}[1] {\textcolor{red}{$\bigstar$ #1 $\bigstar$\\}}
\renewcommand {\comment}[1] {}
\newcommand{\ignore}[1] {}

\usepackage{tikz}

\newcommand{\sm}{\setminus}

\newcommand{\OO}{\mathcal{O}}

\newcommand{\probl}[3]{
\begin{flushleft}
\fbox{
\begin{minipage}{\linewidth}
\noindent {\sc #1}\\
          {\bf Input:} #2\\
          {\bf Output:} #3
\end{minipage}}
\medskip
\end{flushleft}
}

\newtheorem{observation}{Observation}

\newcommand{\mstbl}  {\textsc{MaxSubtreesBL}}

\newcommand{\mc}[1]  {{\sc Max $#1$-Colorable Set}}
\newcommand{\mkc}    {{\mc{k}}}

\newcommand{\cc}   {{\mathcal C}}
\newcommand{\cl}   {{\mathcal L}}
\newcommand{\oo}   {{\mathcal O}}

\newcommand{\TT}   {{\mathcal T}}
\newcommand{\opt}  {\textsc{Opt}}
\newcommand{\vpt}  {\textsc{Vpt}}
\newcommand{\rep}  {\left<T,\TT\right>}

\newcommand{\vect}[1]{\mathbf{#1}}

\newcommand{\load}{load}

\title{Multicast Communications in Tree Networks with Heterogeneous Capacity Constraints}

\author{Yuval Emek}{Department of Industrial Engineering and Management, Technion, Haifa, Israel.}{yemek@technion.ac.il}{https://orcid.org/0000-0002-3123-3451}{}

\author{Shay Kutten}{Department of Industrial Engineering and Management, Technion, Haifa, Israel.}{kutten@technion.ac.il}{}{}

\author{Mordechai Shalom}{TelHai Academic College, Upper Galilee, 12210, Israel,\\
Holon Institute of Technology, Holon,Israel.}{cmshalom@telhai.ac.il}{https://orcid.org/0000-0002-2688-5703}{}

\author{Shmuel Zaks}{Department of Computer Science, Technion, Haifa,  Israel,\\
School of Engineering, Ruppin Academic Center, Israel.}{zaks@cs.technion.ac.il}{https://orcid.org/0000-0001-5637-4923}{}

\authorrunning{Y. Emek et al.}%mandatory. First: Use abbreviated first/middle names. Second (only in severe cases): Use first author plus 'et al.'

%TODO mandatory: Please choose ACM 2012 classifications from https://dl.acm.org/ccs/ccs_flat.cfm 

%TODO mandatory; please add comma-separated list of keywords

\funding{This research was supported in part by PetaCloud - a project funded by the Israel Innovation Authority.}%optional, to capture a funding statement, which applies to all authors. Please enter author specific funding statements as fifth argument of the \author macro.

\acknowledgements{We  thank Eithan Zehavi of Mellanox, Israel, for helpful discussions and providing the simulation data.}%optional

\nolinenumbers %uncomment to disable line numbering

\EventEditors{John Q. Open and Joan R. Access}
\EventNoEds{2}
\EventLongTitle{42nd Conference on Very Important Topics (CVIT 2016)}
\EventShortTitle{CVIT 2016}
\EventAcronym{CVIT}
\EventYear{2016}
\EventDate{December 24--27, 2016}
\EventLocation{Little Whinging, United Kingdom}
\EventLogo{}
\SeriesVolume{42}
\ArticleNo{23}
\begin{document}

\maketitle 
\begin{abstract}
    A widely studied problem in communication networks is that of finding the maximum number of communication requests that can be scheduled concurrently, 
subject to node and/or link capacity constraints. 
In this paper, we consider the problem of finding the largest number of multicast communication requests that can be serviced simultaneously by a network of tree topology,
subject to heterogeneous capacity constraints.
This problem generalizes the following two problems studied in the literature:
a) the problem of finding a largest induced $k$-colorable subgraph of a chordal graph,
b) the maximum multi-commodity flow problem in tree networks.

The problem is already known to be {\nph} and to admit a $c$-approximation ($c \approx 1.58$) in the case of homogeneous capacity constraints.
We first show that the problem is much harder to approximate in the heterogeneous case.
We then use a generalization of a classical algorithm to obtain an $M$-approximation 
where $M$ is  the  maximum  number  of  leaves  of  the  subtrees representing the multicast communications.
Surprisingly, the same algorithm, though in various disguises, is used in the literature at least four times to solve related problems (though the analysis is different).

The special case of the problem where instances are restricted to unicast communications in a star topology network is known to be polynomial-time solvable.
We extend this result and show that the problem can be solved in polynomial time for a set of paths in a tree that share a common vertex.
\end{abstract}

\newpage
\setcounter{page}{1}

\section{Introduction}
\subsection{Motivation and Background}
\runningtitle{Motivation}
Consider a communication graph $G$, in which the vertices represent telephone stations
(trunks) or internet computers and the edges represent connection lines. 
Each station and connection line has a given capacity to handle transmissions. 
Given a set of transmission requests as paths between vertices, the problem is to find the maximum set of transmissions
(paths in $G$) which can be handled at any moment. 
Thus, we must find a maximum set of paths that do not overload the capacity of the vertices or edges.

This scenario occurs in numerous applications in communication networks, and it can be interpreted also in the context of a production line, where vertices are associated with machines, and bounds with number of units that can be produced by a machine during one time unit. 
In these applications, the aim is to maximize the number of production lines that can be active simultaneously. 
Such algorithms can be used in successive steps, so as to schedule all of the given requests. 

Stated in graph-theoretical terms, this is a special case of the  fundamental class of graph optimization problems, in which  the objective is to find a maximum number of given subgraphs under a given set of constraints. 
Specifically, when the communication graph $G$ is a tree and the given subgraphs are subtrees of $G$, the problem generalizes the maximum $k$-colorable set problem in chordal graphs.
In that problem, one looks for the maximum number of vertices of a given graph that can be colored with $k$ colors (such that no two vertices of the same color are adjacent). 

\runningtitle{Graph classes and subtree representations}
In this paper, we mention the following classes of 
graphs that can be represented as intersections of subtrees of a tree.
\begin{itemize}
\item {\em Chordal graphs}: a graph is \emph{chordal} if it does not contain induced cycles of length four or more.
In other words, in a chordal graph, every cycle with more than three vertices contains a chord (i.e., an edge that joins two non-adjacent vertices of the cycle).
It is known that a graph is chordal if and only if it is the vertex-intersection graph of subtrees of a tree ~\cite{Gavril74}.
It is also known that a graph is chordal if and only if it has a perfect elimination order of its vertices.
Such an order corresponds to the traversal of the representation tree in a bottom up manner and each time considering the subtrees that have the current vertex as their root.
It is also known that the chromatic number of a chordal graphs (the graph being perfect) equals its clique number and equals the maximum number of subtrees in the representation that share a common vertex.
We refer the reader to \cite{Golumbic:2004:AGT:984029} for more details on chordal graphs.
\item {\em {\vpt} graphs}:
Such a graph is the vertex-intersection graph of paths in a tree. 
In other words, a {\vpt} graph is a chordal graph that has a representation in which every subtree is a path.
\item {\em Directed path graphs}: Such a graph is the vertex-intersection graph of (directed) paths in a rooted tree.
We emphasize that a directed path graph is an undirected graph. 
In fact, it is easy to see that it is a {\vpt} graph.
\item {\em Interval graphs}: an \emph{interval graph} is the intersection graph of intervals on the real line. 
Clearly, an interval graph is a directed path graph.
\end{itemize}

\subsection{Related Work and Our Contribution}
\runningtitle{Homogeneous capacities}
The maximum independent set problem on an interval graph is equivalent to the problem of finding a maximum cardinality subset of a given set of intervals such that no pair of them intersect.
Such a set of intervals can be found using the Earliest Deadline First (EDF) algorithm \cite{Stankovic1998}.
This algorithm processes the intervals greedily in the right-endpoint order, and 
an interval is added to the solution whenever it does not intersect an interval already in the solution, or in other words, whenever its addition to the solution would preserve feasibility.

An extension of this problem is the maximum independent set problem in chordal graphs. 
A polynomial-time greedy algorithm for this problem is presented by Gavril \cite{Gavril72AlgorithmsForChordalGraphs}.
In fact, this algorithm is a generalization of the EDF algorithm:
vertices are processed in a perfect elimination order, and a vertex is added to the solution as long as it preserves feasibility.
Recall that a perfect elimination order corresponds to a traversal of the tree in a bottom-up manner.
Clearly, when the tree is a path, this order becomes the right-endpoint order.

Another extension is the maximum $k$-colorable set problem in interval graphs, which is the problem of finding a maximum cardinality subset of vertices that can be partitioned into $k$ independent sets.
A polynomial-time greedy algorithm for this problem is presented by Gavril and Yannakakis \cite{GY87}.
The algorithm can be stated in the same way:
"process the intervals in the right-endpoint order and add it to the solution if it preserves feasibility".
The same algorithm was suggested also in \cite{carlisle1995k}.
In this work, Gavril and Yannakakis showed that the last two extensions above, applied together, make the problem {\nph}.
Namely, to find a maximum $k$-colorable subgraph of a chordal graph is {\nph}.
This problem is equivalent to the problem of finding a maximum number of subtrees among a given set of subtrees of a tree such that 
every vertex is contained in at most $k$ of them.
It is worth noting that this is the homogeneous variant of our problem in which all the vertices (and edges) have the same capacity of $k$ subtrees.

Chakaravarthy and Roy \cite{ChakaravartyRoy09MaxkColorableChordal} suggested a 2-approximation algorithm for the weighted version of the problem, in which every subtree has a weight and the value of a solution is the sum of the weights of its subtrees.
Their algorithm is a two pass algorithm where the first pass suffices to solve the unweighted variant of the problem.
As they point out, this first pass is a greedy algorithm that processes the vertices in a perfect elimination order, and if run with $k=1$ it reduces to Gavril's algorithm for maximum independent sets.
Putting differently, the algorithm is again: "process vertices in a perfect elimination order and add a vertex to the solution if it preserves feasibility".
However, a different algorithm with a better approximation ratio, namely $1.582$, was proposed earlier by Wan and Liu \cite{WanLiu98MaximumThroughputinWRO}.
Though they did not consider the same problem explicitly, they showed that for every graph class for which the maximum independent set problem can be solved in polynomial time,
the maximum $k$-colorable subgraph problem can be approximated with an approximation ratio of $\frac{e}{e-1} \approx 1.582$.

\runningtitle{Heterogeneous capacities}
In ~\cite{GargVY97} Garg, Vazirani and Yannakakis studied the maximum integral multi-commodity flow problem in tree networks. 
In this problem, we are given pairs of vertices of a tree with capacities on its edges.
The goal is to find a maximum number of pairs of vertices (i.e., paths of the tree) such that the number of paths that use an edge does not exceed its capacity.
They provide an algorithm to solve the problem and another one to solve its dual, 
and they show that the solutions are at most a factor of 2 away from each other,
thus providing a primal-dual proof of the algorithm being a $2$-approximation.
The algorithm that solves the primal problem is in fact the same greedy algorithm: 
"process the tree in some bottom-up order, at every vertex consider all the paths whose highest vertex is this one and a path to the solution if it preserves feasibility".
In their work, they also showed that the maximum multi-commodity flow problem is {\maxsnph} for general tree networks.

\runningtitle{Our contributions:}
In the current paper, we first show that the heterogeneous version of the problem is much harder to approximate when one considers subtrees instead of paths.
Specifically, it cannot be approximated to within $n^{1-\epsilon}$ for any $\epsilon > 0$ where $n$ is the number of subtrees.
On the other hand, when the subtrees are claws (i.e. sub-stars with 3 leaves) and all the loads are $1$ except for the center vertex, the problem becomes $\apxh$.
Then we use the same greedy algorithm mentioned above to obtain an $M$-approximation, 
where $M$ is the maximum number of leaves of the trees in the representation, and the root of a tree is not counted as a leaf.
We provide a direct proof of this fact, i.e. one that does not use duality theory, which is thus different from the above-mentioned algorithm of~\cite{GargVY97}.
Besides considering trees (instead of paths), we allow for a  demand to be specified per subtree (i.e., per multicast) basis.
This result implies a $2$-approximation for the maximum multi-commodity flow problem with demands and an optimal algorithm for the maximum $k$-colorable subgraph problem in directed path graphs. 
Figure \ref{fig:GreedyAlgorithmEvolution} summarizes the above discussion.

\begin{figure}
\begin{center}
\includegraphics[width=\textwidth]{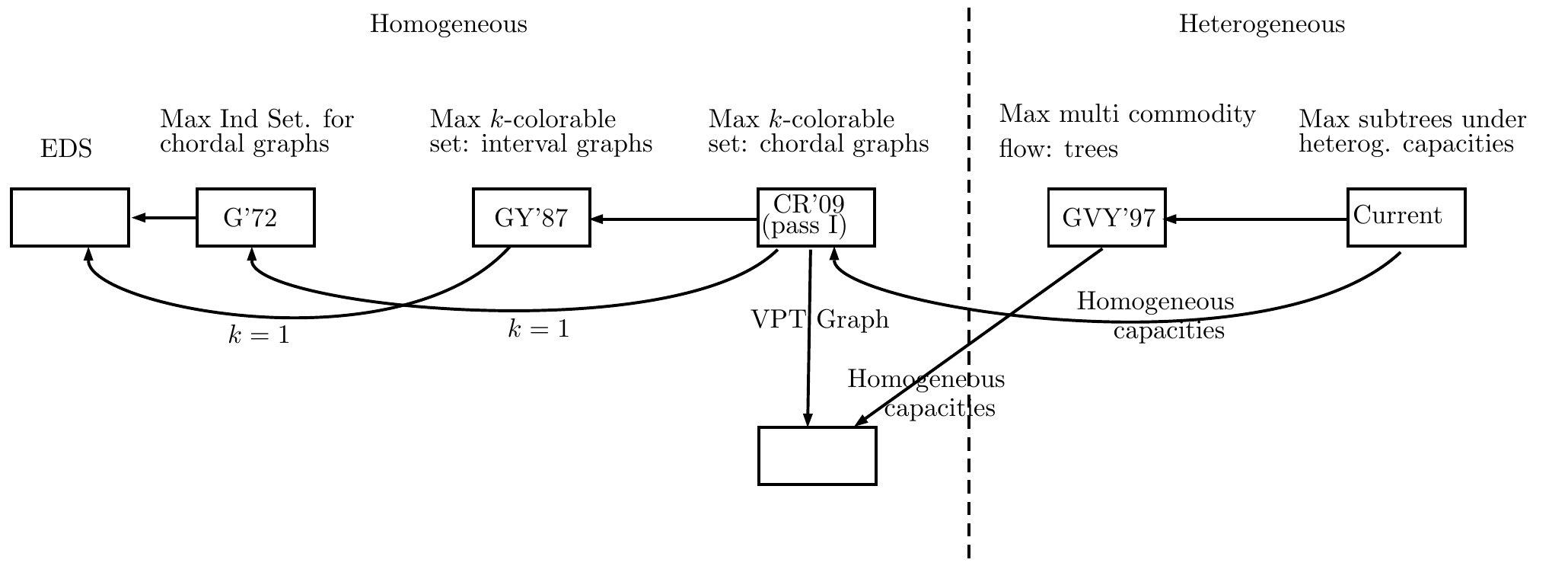}
\caption{
Applicability of variants of the greedy algorithm to various problems. An arrow from rectangle $A$ to rectangle $B$ indicates that the algorithm in $A$ can be used to solve the problem in $B$. The description on the arrow indicates inputs that algorithm in $A$ might need in addition to an instance of problem $B$.
Whenever algorithm $A$ can be applied to problem $B$, if ties are broken in the same way (namely if they use the same perfect elimination ordering/bottom-up traversal), algorithm $A$ simulates the run of the algorithm in $B$, and produces the same output in the same order.
} \label{fig:GreedyAlgorithmEvolution}
\end{center}
\end{figure}

In \cite{GargVY97}, it is shown that the problem of finding the largest number of paths in a star network under heterogeneous capacity constraints can be solved in polynomial time.
We also extend this result and show that the problem is polynomial-time solvable for any tree, provided that the paths share a common a vertex.
In addition to the rigorous analysis, we tested our greedy algorithm on random data and real-life data.
We found out that its performance is very close to optimum.

\runningtitle{Applications}
In \cite{halldorsson2003sum}, an optimal algorithm that solves of the maximum $k$-colorable subgraph problem optimally is used
as a subroutine to get an approximation algorithm for the minimum sum coloring problem in interval graphs.
This problem corresponds to the scheduling of items to minimize average completion time.
In \cite{misra2019parameterized}, it was used in the context of parameterized algorithms. 
Examples of works in which the maximum $k$-colorable set problem is used in applications are \cite{erlebach1998maximizing}, 
in the context of optical networks with wavelength division multiplexing (WDM), where multiple connections can share a link if they are transmitted on different wavelengths, and the aim is to satisfy a maximum number of connection requests in a tree network; \cite{bermond2006optimal}, where it was used 
in the context of traffic grooming in optical networks; and in \cite{krause2013optimal},
in the context of register allocation, where it  corresponds to the case where the input programs are in static single assignment form.

\runningtitle{Paper structure:}
In Section \ref{sec:prelim}, we present basic terms and notations, formally present the optimization problems and provide hardness results.
In Section \ref{sec:mstbl}, we present our algorithms and their analysis.
We summarize the theoretical results, present simulation results and discuss open problems in Section \ref{sec:conc}.
 
\section{Preliminaries}\label{sec:prelim}
\runningtitle{Sets and Multisets}
For two non-negative integers $n_1,n_2$, denote by  $[n_1,n_2]$ 
the set of integers that are not smaller than $n_1$ and not larger than $n_{2}$.
We also use $[n]$ as a shorthand for $[1,n]$.
Given a multiset $S$ and an element $e$, we denote by $m_S(e)$ the \emph{multiplicity} of $e$ in $S$,
i.e., the number of times $e$ appears in $S$.
Clearly, $m_S(e) \neq 0$ if and only if $e \in S$.
Also, $S' \subseteq S$ if and only if $m_{S'}(e) \leq m_S(e)$ for every $e$.
We denote by $\hat{S}$ the set of distinct elements of $S$.
We say that a multiset $S$ is given in \emph{compact} form if it is given by the set $\hat{S}$ of its distinct elements and their associated multiplicities.
  
\runningtitle{Graphs}
We use standard terminology and notation for graphs, see for instance \cite{D12}.
Given a simple undirected graph $G$, we denote by $V(G)$ the set of vertices of $G$ and by $E(G)$ the set of the edges of $G$.
Denote an edge between two vertices $u$ and $v$ as $uv$.
We say that the edge $uv \in E(G)$ is {\em incident} to $u$ and $v$, 
$u$ and $v$ are the endpoints of $uv$, 
and $u$ is adjacent to $v$ (and vice versa).
We use the term \emph{object} to refer to a vertex or an edge.

A \emph{coloring} of a graph $G$ is a function from $V(G)$ into the set of positive natural numbers.
A coloring of $G$ is \emph{proper} if $u$ and $v$ are assigned different colors by the coloring whenever $uv$ is an edge of $G$.  
A graph is $k$-colorable if it has a proper coloring of its vertices using colors from $[k]$.
The chromatic number of a graph $G$, denoted by $\chi(G)$, is the smallest integer $k$ such that $G$ is $k$-colorable. 
The clique number $\omega(G)$ of a graph $G$ is the size of its largest clique.

\runningtitle{Trees, subtrees and their intersection graphs}
Let $U$ be a set of subtrees of a tree $T$, and $o$ an object of $T$.
Denote as $U_o$ the set of subtrees in $U$ that contain the object $o$,
and by $\load(U,o)$ their number $\abs{U_o}$ which is termed as the \emph{load} of $o$ induced by $U$. 
Denote by $\vect{\load(U)}$ the vector consisting of the values $\load(U,o)$ indexed by the objects of $T$,
and by $L_V(U)$ (resp. $L_E(U)$) the maximum of $\load(U,o)$ over all vertices (resp. edges) $o$ of $T$.

Let $\TT = \set{ T_1, \ldots, T_n}$ be a set of subtrees of a tree $T$.
Let $G=(V,E)$ be a graph over the vertices $\set{v_1, \ldots, v_n}$ such that $v_i$ and $v_j$ are adjacent if and only if $T_i$ and $T_j$ have a common vertex.
Then, $G$ is termed as the \emph{vertex-intersection graph} of $\rep$ and conversely $\rep$ is termed a subtree intersection representation, 
or simply a \emph{representation} of $G$. 
Whenever there is no confusion, we will denote the representation $\rep$ simply as $\TT$.

\runningtitle{Graph classes} 
We have already introduced the following graph classes: Chordal graphs, $\vpt$ graphs, directed path graphs, and interval graphs.
Let us now mention some of their properties needed for this paper. 

Let $G$ be a chordal graph with a representation $\rep$, and $v$ a vertex of $T$.
Clearly, the set $\TT_v$ of subtrees in $\TT$ that contain the vertex $v$ corresponds to a clique of $G$.
It is also known that there is a representation (where $T$ is termed the clique-tree of $G$) such that every maximal clique of $G$ corresponds to $\TT_v$ for some vertex $v$ of $T$. 
Therefore, the clique number $\omega(G)$ of $G$ is equal to $L_V(\TT)$.
Since chordal graphs are perfect, we have $\chi(G)=\omega(G)$.
It is known that one can find a representation of a given chordal graph $G$ in polynomial time, and such a representation $\rep$ can be used to determine the chromatic number of $G$, since $\chi(G)=\omega(G)=L_V(\TT)$.

\runningtitle{Problem Statements} 
\probl
{Maximum $k$-Colorable Set (\mkc)}
{A graph $G$, a positive integer $k$.}
{A $k$-colorable induced subgraph of $G$ of maximum order.}

Since the decision version of the chromatic number problem is equivalent to asking whether the maximum $k$-colorable set is the entire vertex set, we have the following observation.
\begin{observation}\label{obs:ChromaticNPH}
If the chromatic number problem is $\nph$ when restricted to a class of graphs $\cc$, then \mkc, when restricted to $\cc$, is also $\nph$. 
\end{observation}

Recall that a representation $\rep$ of a chordal graph $G$ can be found in polynomial time 
and a set of vertices of $G$ is $k$-colorable if and only if the corresponding set $U \subseteq \TT$ of subtrees satisfies $\load(U,v) \leq k$ for every node $v$ of $T$. 
Therefore, we define the following problem that generalizes the {\mkc} problem in chordal graphs.
\newcommand{\inst}{(T, \TT, \vect{k})}

\probl
{Maximum Subtrees with Bounded Load(\mstbl)}
{
A triple $\inst$ where \\
$T$ is a tree, $\TT$ is a set of subtrees of $T$, and
$\vect{k}$ is a vector of non-negative integers indexed by the objects of $T$.
}
{A  maximum cardinality subset $U$ of $\TT$ such that $\vect{\load(U)} \leq \vect{k}$.}

Given an instance $\inst$ of $\mstbl$, a set $U \subseteq \TT$ of subtrees and an object $o$ of $T$, 
we say that $o$ is \emph{overloaded} (resp. \emph{tight}) in $U$ if $\load(U,o) > k_o$ (resp. $\load(U,o) = k_o$).
A set $U \subseteq \TT$ of subtrees of $T$ is \emph{feasible} (for the instance $\inst$) if no object of $T$ is overloaded in $U$.

Throughout this paper, $T$ is a tree with an arbitrary vertex $r$ chosen as its root
and $\TT=\set{T_1, \ldots, T_n}$ is a set of subtrees of $T$.
For every subtree $T_i \in \TT$ of $T$ we denote by $root(T_i)$ its (unique) vertex that is closest to $r$,
and by $leaves(T_i)$ the number of leaves of $T_i$ that are different from $root(T_i)$.
Finally, $M \defined \max_{T_i \in \TT} leaves(T_i)$ is the maximum number of leaves of a subtree, except its root.

An instance of the maximum multi-commodity flow problem, 
consists of graph $G$ with a capacity function $\kappa$ on its edges 
and triples $(s_i, t_i, d_i)$ where $s_i, t_i$ are vertices of $G$ and $d_i$ is a positive \emph{demand} of some commodity to be sent from $s_i$ to $t_i$ through the edges of $G$.
The output is a set of flow functions $f_i$ on the edges of $G$ such that $f_i$ is a flow from $s_i$ to $t_i$ of value $F_i$.
The goal is to maximize $\sum_i F_i$.
When $G$ is a tree, since there is a unique path between $s_i$ and $t_i$ on $G$, we can consider the instance as a set of paths $\set{P_1, P_2, \ldots}$ with a demand $d_i$ associated with every path $P_i$.
For the same reason, the function $f_i$ is equal to $F_i$ in every edge on the path between $s_i$ and $t_i$, and zero elsewhere. 
We can also assume without loss of generality that every $d_i$ is an integer.
Therefore, the maximum multi-commodity flow problem is a subproblem of {\mstbl} where the set $\TT$ is a multiset of paths given in its compact form.
By this equivalence, the hardness results in \cite{GargVY97} imply

\begin{theorem}
{\mstbl} is {\nph} and {\maxsnph} even when $M=2$ and $k_e \in \set{1,2}$ for every edge $e$ of $T$, and $k_v=\infty$ for every vertex $v$ of $T$.
\end{theorem}

We provide further hardness results in the following theorem.
\begin{theorem}
~\\
\vspace{-4mm}
\begin{enumerate}
    \item $\mstbl$ is not approximable within $\abs{\TT}^{1-\varepsilon}$ for every $\varepsilon > 0$ even when 
    \begin{itemize}
        \item $T$ is a star, and
        \item $k_o=1$ for every object $o$ of $T$ except its center,
    \end{itemize} 
    \item $\mstbl$ is $\apxh$ even when 
    \begin{itemize}
        \item $T$ is a star,
        \item $k_o=1$ for every object $o$ of $T$ except its center, and
        \item $M=3$.
    \end{itemize} 
\end{enumerate}
unless $\nohyp$.

\end{theorem}
\begin{proof}
By approximation preserving reduction from the maximum independent set problem which is known not to be approximable within $n^{1-\varepsilon}$ for any $\varepsilon > 0$ \cite{H99}, where $n$ is the number of vertices of the graph,
and also $\apxh$ for graphs with maximum degree three \cite{PY91}.
Given a graph $G$ on $n$ vertices which is an instance of the maximum independent set problem,
we construct a star $T$, every edge of which corresponds to an edge of $G$.
Now a vertex $v$ of $G$ corresponds to a sub-star $T_v$ of $T$ that consists of the edges corresponding to the edges of $G$ incident to $v$. 
Then $G$ is the edge-intersection graph of the sub-stars $\TT = \set{T_v | v \in V(G)}$.
We note that the degree of a vertex $v$ of $G$ is equal to the number of leaves of the sub-star $T_v$.
A set of vertices of $G$ is an independent set if and only if they are pairwise non-adjacent 
if and only if the corresponding sub-stars of $T$ are pairwise non-edge-intersecting  
if and only if they induce a load of at most 1 on every object except the center of $T$. 
Therefore, $G$ has an independent set of $k$ vertices if and only if the $\TT$ contains $k$ sub-stars that induce load of at most 1 on the leaves of $T$.
We conclude that this reduction preserves approximation ratio.
\end{proof} 

\comment{We mentioned the following two problems too
\probl
{Problem 1}
{$(T_i, d_i)$}
{$x_i \in [0..1]$ s.t. $\sum x_i d_i$ is maximum}
}

\comment{
\probl
{Problem 2}
{$(T_i, d_i)$}
{$x_i \in [0..1]$ s.t. $\sum x_i$ is maximum}
}

\section{Maximum Number of Subtrees with Bounded Load}\label{sec:mstbl}
In this section, we present a greedy algorithm for the $\mstbl$ problem and prove that it is an $M$-approximation.
This implies a) a $2$-approximation whenever the subtrees are paths, and 
b) optimality whenever the tree is directed according to some root and all the paths follow the direction.
We then consider the special case of paths sharing a common vertex and provide an optimal algorithm for this case, 
thus generalizing a result presented in \cite{GargVY97}.

\newcommand{\alg}{\textsc{BottomUpGreedy}}
\newcommand{\greedy}{\textsc{BottomUpGreedy}}

\alglanguage{pseudocode}

\begin{algorithm}[htbp]
\caption{$\alg$}\label{alg:BottomUpGreedy}
\begin{algorithmic}[1]
\Require {An instance $\inst$ of $\mstbl$}
\Ensure{Return a subset $U$ of $\TT$ such that $\vect{\load(U)} \leq \vect{k}$.}
\Statex
\State $U \gets \emptyset$. 
\For{every vertex $v$ in some post-order traversal of $T$}
    \For {every subtree $T_i \in \TT$ such that $root(T_i) = v$} 
        \If {$U + T_i$ is feasible}
            \State $U \gets U + T_i$.
        \EndIf
    \EndFor
\EndFor
\State \Return $U$.

\end{algorithmic}
\end{algorithm}

\newcommand{\sol}{ALG}

Assume without loss of generality, that the subtrees $\TT$ are numbered in the order they are visited by $\alg$,
i.e., in some bottom-up order of their roots in $T$,   
and let $\TT_i$ be the first $i$ subtrees of $\TT$ in this order.
Let $\opt$ be some optimal solution of {\mstbl}, i.e. a set of subtrees from $\TT$ of maximum size such that $\vect{\load(\opt)} \leq \vect{k}$.
Let $\sol$ denote a set of subtrees returned by the algorithm on input $\inst$, and let $\sol_i = \sol \cap \TT_i$.

\begin{lemma}\label{lem:SetsS}
There exists a sequence $S_0, S_1, \ldots, S_n$ of feasible subsets of $\TT$ such that for every $i \in [0,n]$ we have
\begin{enumerate}
\item $S_i \sm \TT_i \subseteq \opt$,
\item $\abs{\opt} - \abs{S_i} \leq (M-1) \abs{\sol_i}$, and
\item $S_i \cap \TT_i = \sol_i$.
\end{enumerate}
Note that the first and last conditions above, together imply that $S_i \setminus \sol_i$ is contained in $\opt$.
\end{lemma}

\begin{proof}
We define $S_i$ inductively, and prove by induction on $i$.
We start with $S_0 = \opt$. 
Clearly,  $S_0$ is feasible, and it is easy to verify that it satisfies all three conditions. 
For $i \in [n]$ we proceed as by dividing into cases.

\begin{itemize}
\item{$T_i \notin \sol$ and $T_i \notin S_{i-1}$ or $T_i \in \sol$ and $T_i \in S_{i-1}$:} in this case we define $S_i = S_{i-1}$ which is feasible by the inductive hypothesis.
Clearly, $S_{i-1} \cap \set{T_i} = \sol \cap \set{T_i}$ in both cases.
We now verify that all the conditions are satisfied by $S_i$.
\begin{enumerate}
\item $S_i \sm \TT_i = S_{i-1} \sm \TT_i = S_{i-1} \sm \TT_{i-1} - T_i \subseteq  S_{i-1} \sm \TT_{i-1} \subseteq \opt$.
\item $\abs{\opt} - \abs{S_i} = \abs{\opt} - \abs{S_{i-1}} \leq (M-1) \abs{\sol_{i-1}} \leq (M-1) \abs{\sol_i}.$
\item $\begin{array}{lcl}
S_i \cap \TT_i & =  & S_{i-1} \cap \TT_i = S_{i-1} \cap (\TT_{i-1} + T_i) =  (S_{i-1} \cap \TT_{i-1}) \cup (S_{i-1} \cap \set{T_i})\\
& = & ALG_{i-1} \cup (\sol \cap \set{T_i}) = \sol_i
\end{array}$.
\end{enumerate}

\item{$T_i \notin \sol$ and $T_i \in S_{i-1}$:} in this case we have
\[
S_{i-1} = (S_{i-1} \cap \TT_{i-1}) \cup (S_{i-1} \sm \TT_{i-1}) \supseteq (S_{i-1} \cap \TT_{i-1}) + T_i = \sol_{i-1} + T_i.
\]
Since, by the inductive hypothesis, $S_{i-1}$ is feasible, we conclude that $\sol_{i-1} + T_i$ is also feasible.
Then, by the greedy behaviour of the algorithm, we have $T_i \in \sol$, a contradiction.

\item{$T_i \in \sol$ and $T_i \notin S_{i-1}$:} We further divide this case into two subcases.
\begin{itemize}
\item{$S_{i-1} + T_i$ is feasible:}
in this case we define $S_i = S_{i-1} + T_i$ and verify that all conditions are satisfied.
\begin{enumerate}
\item $S_i \setminus \TT_i = (S_{i-1} + T_i) \setminus \TT_i = S_{i-1} \setminus \TT_{i-1} \subseteq \opt$.
\item $\abs{\opt} - \abs{S_i} \leq \abs{\opt} - \abs{S_{i-1}} \leq (M-1) \abs{\sol_{i-1}} < (M-1) \abs{\sol_i}$.
\item $S_i \cap \TT_i = (S_{i-1} + T_i) \cap (\TT_{i-1} + T_i) = (S_{i-1} \cap \TT_{i-1}) + T_i =  \sol_{i-1} + T_i = \sol_i$.
\end{enumerate}

\item{$S_{i-1} + T_i$ is not feasible.}
Consult Figure \ref{fig:TreeIntersection} for the following discussion.
Since $S_{i-1}$ is feasible, but $S_{i-1} + T_i$ is not,
$T$ contains a non-empty set of objects that are tight in $S_{i-1}$, but overloaded in $S_{i-1} + T_i$. 
Let $O = \set{ o \in V(T) \cup E(T)|~\load(S_{i-1} + T_i,o) = k_o + 1}$ be the set of these objects.
Clearly, $O$ is contained in $T_i$.
Consider the $leaves(T_i)$ paths connecting $root(T_i)$ to the leaves of $T_i$.
Let $\oo$ be the subset of objects in $O$ each of which is the last object in $O$ (i.e. farthest from $root(T_i)$) in one of these root-to-leaf paths.
Clearly, $\abs{\oo} \leq leaves(T_i) \leq M$.

\begin{figure}
\begin{center}
\includegraphics{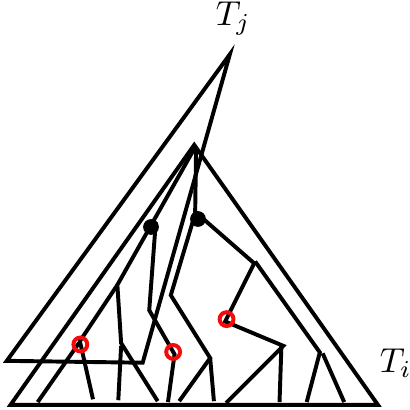}
\caption{
The intersection of two subtrees $T_i, T_j \in \TT$ with $i \leq j$. 
The set of overloaded objects in the example consists of five vertices (i.e. $\abs{O}=5$) emphasized with circles. 
The objects drawn with red circles are the elements of $\OO$ which consists of the lowest elements of $O$ (i.e. closest to the leaves of $T$). 
The tree $T_j \in \TT'$ contains one such element, and all its ancestors (including those in $O$). 
In this way the set $\TT'$ contains at most $\abs{\oo}$ subtrees and covers all the elements of $O$.}\label{fig:TreeIntersection}
\end{center}
\end{figure}

We have $\sol_i = \sol_{i-1} + T_i = (S_{i-1} \cap \TT_{i-1}) + T_i$, and this set is feasible.
Recall that $S_{i-1} + T_i$ is infeasible.
Therefore, every object of $\oo$, is contained in at least one subtree of $S_{i-1} \sm \TT_i$.
We construct a set $\TT'$ of at most $\abs{\oo} \leq M$ subtrees by picking one subtree from $S_{i-1} \sm \TT_i$ that contains $o$ for every object $o \in \oo$.
The order the algorithm visits the subtrees implies
that $root(T_j)$ is a (not necessarily proper) ancestor of $root(T_i)$ for every $T_j \in \TT'$.
Therefore, $\TT'$ hits all the overloaded objects $\oo$ (of $S_{i-1} + T_i$).
Then no object of $T$ is overloaded in $(S_{i-1} + T_i) \sm \TT'$. 
Therefore, $(S_{i-1} + T_i) \sm \TT'$ is feasible.
We define $S_i = (S_{i-1} + T_i) \sm \TT'$ and verify that the conditions are satisfied.
\begin{enumerate}
\item $S_i \sm \TT_i \subseteq (S_{i-1} + T_i) \sm \TT_i = S_{i-1} \sm \TT_{i-1} \subseteq \opt$.
\item $\abs{\opt} - \abs{S_i} \leq \abs{\opt} - (\abs{S_{i-1}} + 1 - \abs{\TT'}) = \abs{\opt} - \abs{S_{i-1}} + \abs{\TT'} - 1 \leq (M-1) \abs{\sol_{i-1}} + M - 1 = (M-1) \abs{\sol_i}$.
\item $S_i \cap \TT_i = (S_{i-1} + T_i \setminus \TT') \cap \TT_i = (S_{i-1} \cap \TT_{i-1}) + T_i = \sol_{i-1} + T_i = \sol_i$ where the second equality holds because $\TT' \cap \TT_i = \emptyset$.
\end{enumerate}
\end{itemize}
\end{itemize}
\end{proof}

\begin{theorem}\label{thm:MApprox}
The approximation ratio of $\alg$ is exactly $M$.
\end{theorem}
\begin{proof}
Let $S_0, \ldots, S_n$ be the sets whose existences are guaranteed by Lemma \ref{lem:SetsS}.
Then, $S_n = S_n \cap \TT_n = \sol_n = \sol$, and we conclude
\[
\abs{\opt} - \abs{\sol} = \abs{\opt} - \abs{S_n} \leq (M-1) \abs{\sol_n} = (M-1) \abs{\sol}. \]
It remains to show that the approximation ratio of $\alg$ is at least $M$.
Consider the instance $\inst$ of $\mstbl$ where $T$ is a star with center $0$ and $M$ leaves $[M]$,
$\TT=\set{P_{0,j} \vert j \in [M]} \cup \set{T}$ where $P_{i,j}$ denotes the unique path between $i$ and $j$ in $T$.
Let also $k_0 = M$ and $k_o=1$ for every other object $o$.
$\alg$ processes the paths of $\TT$ in an arbitrary order since $root(T_i)=0$ for every $T_i \in \TT$. 
The algorithm may consider first the star $T$ and return $\set{T}$ as a solution. 
On the other hand, $\set{P_{0,j} \vert j \in [M]} \cup \set{T}$ is a feasible solution with $M$ subtrees.
\end{proof}

We now note that our algorithm, as described, is polynomial when $\TT$ is a set, but pseudo polynomial if $\TT$ is a multiset given in compact form.
We now present a simple modification to Algorithm {\greedy} so that it runs in polynomial time also in the latter case.
Suppose that each subtree $T_i \in \TT$ has multiplicity $d_i$.
Note that when a tree $T_i$ is considered, in order to decide whether $U+T_i$ is feasible we have to compute the minimum remaining capacity of the objects in $T_i$.
Let $c$ be this capacity. 
Then $T_i$ is be added to $U$ with multiplicity $\min \set{d_i, c}$.
After this modification, ${\greedy}$ considers every tree in $\hat{\TT}$ exactly once, regardless from its multiplicity in $\TT$.
Therefore, we have

\begin{theorem}
There is an $M$-approximation greedy algorithm for $\mstbl$ that runs in polynomial time also when the set $\TT$ is given in compact form.
\end{theorem}

As remarked in \cite{GargVY97}, when $M=2$ and $T$ is a star, the problem is equivalent to the $b$-matching problem, thus polynomial-time solvable.
In the sequel we generalize this result.

\begin{theorem}\label{thm:HieararchicalBMatching}
{\mstbl} can be solved in polynomial time when $M=2$ and $T$ contains a vertex $r$ such that when $T$ is rooted from $r$, every path in $\TT$ is either directed or contains $r$.
\end{theorem}
\begin{proof}
We first show that under our assumption, {\mstbl} can be solved in polynomial time if and only if {\mstbl} can be solved in polynomial time for instances where $M=2$ and 
there is a vertex $r$ of $T$ contained as internal vertex by all the paths.
The only if part being trivial, we show the if part.

Consider a variant of {\greedy} where at every vertex $v \in V(T)$ the subtrees rooted at $v$ are processed in the order of their number of leaves except $v$. 
Consider an execution of this variant with some postorder traversal of $T$ ending at the root $r$,
and let $T_i$ be the last subtree rooted at $r$ that has at most one leaf except $r$.
Then $\TT_i$ (resp. $\TT \setminus \TT_i$) is the set of paths that do not (resp. do) contain $r$ as an internal vertex.
By our assumption, for every tree $T_j \in \TT_i$ the number of leaves of $T_j$ except its root is at most 1.
Then Lemma \ref{lem:SetsS} holds for every $j \leq i$ (and in particular for $i$) with $M=1$.
Namely, we have $S_i \cap \TT_i = \alg_i$ and $\abs{S_i}=\abs{\opt}$.
In other words, $S_i$ is an optimal solution that contains $\alg_i$.
Adding to $\alg_i$ and optimal solution of to solve the instance obtained by the removal of $\alg_i$ we get an optimal solution. 

We now show that instances with $M=2$ and all the paths having a vertex $r$ in common can be solved in polynomial time.
Without loss of generality we can assume that the endpoints of the paths of $\TT$ are leaves of $T$.
Indeed, if this is not the case for some endpoint $u$ of some path of $\TT$ one can add to $T$ a new leaf $v$ adjacent to $u$ and add the edge $uv$ to $P$.
Finally, setting the capacities of the new objects $v$ and $uv$ to infinity we get an equivalent instance.
For a vertex $v$ of $T$, and let $\cl_v$ be the set of leaves of $\TT$ that are descendants of $v$, and if $v \neq r$ let $e_v$ be the unique edge between $v$ and its parent.
Let $G$ be the multi-graph with vertex set $\cl_r$ (i.e. the set of leaves of $T$) such that $uv$ is an edge of $G$ if and only if there is a path between $u$ and $v$ in $\TT$.
Clearly, a subset $U$ of paths of $\TT$ corresponds to a set of edges of $G$, 
in other words, to a subgraph $H$ of $G$.
Consider a subgraph $H$ of $G$ that corresponds to a feasible solution $U$.
The degree $d_H(v)$ of a vertex $v$ of $H$, is the number of paths of $U$ with an endpoint in the leaf $v$ of $T$.
Since $U$ is feasible, we have $d_H(v) \leq \min{k_v, k_{e_v}}$.
Now let $v$ be any vertex of $T$.
A path of $\TT$ has an endpoint in $\cl_v$ if and only if it contains $v$ as an internal vertex. 
Furthermore, if $v \neq r$ we have that a path of $\TT$ has an endpoint in $\cl_v$ if and only if it contains $e_v$.
Then $H$ is a subgraph of $G$ corresponding to a feasible solution of our instance if and only if
\begin{equation}\label{eqn:degreeBounds}
\sum_{\ell \in \cl_v} d_H(\ell) \leq b_v
\end{equation}
where
\[
b_v = \left\{
\begin{array}{ll}
    k_v & \textrm{if~} v=r \\
    \min \set{k_v, k_{e_v}} & \textrm{otherwise}.
\end{array}
\right.
\]
We note that for every pair $u,v$ of vertices of $T$ either the sets $\cl_u$ and $\cl_v$ are either disjoint or one is contained in the other.
Such a set system is termed \emph{laminar}.
Then our objective becomes to find a subgraph $H$ of $G$ with maximum number of edges subject to the degree constraints \eqref{eqn:degreeBounds} where the set system $\set{\cl_v: v \in V(T)}$ is laminar.
This is an extension of the classical $b$-matching problem and it can be solved in polynomial time. (see \cite{EKSZ19-HierarchicalBMatching-Arxiv}).
\footnote{A pseudo-polynomial solution for the weighted case is provided in \cite{KaparisL2014OnLaminarMatroidsAndBMatchings}}. 
\end{proof}

The algorithm implied by the proof of Theorem \ref{thm:HieararchicalBMatching} is provided in Algorithm \ref{alg:VPS}. 
We finally note that this algorithm too, though pseudo-polynomial as described, 
can easily be modified to run in polynomial-time.

\alglanguage{pseudocode}
\begin{algorithm}
\caption{Algorithm using Hierarchical $b$-matching}\label{alg:VPS}
\begin{algorithmic}[1]
\Require {An instance $\inst$ of $\mstbl$}
\Require {$\overrightarrow{T}$ is a rooted tree obtained from $T$ with root $r$}
\Require {$\TT$ is a set of paths}
\Require {If $P \in \TT$ does not contain $r$ as an internal vertex, it is a directed path of $\overrightarrow{T}$}.
\Ensure {Return a subset $U$ of $\TT$ such that $\vect{\load(U)} \leq \vect{k}$.}
\Statex
\State $U \gets \emptyset$. 
\For{every vertex $v$ in some post-order traversal of $T$ ending at $r$}
    \For {every path $T_i \in \TT$ such that $root(T_i) = v$ and $T_i$ is a directed path of $\overrightarrow{T}$} 
        \If {$U + T_i$ is feasible}
            \State $U \gets U + T_i$.
        \EndIf
    \EndFor
\EndFor
\Statex
\State $\TT \gets \TT \setminus U$. 
\For {every object $o$ of $T$}
\State $k_o \gets k_o - \load(U,o)$
\EndFor
\Statex
\For{every non-leaf $u$ of $T$ that is an endpoint of some path in $\TT$}
\State Add a new vertex $v$ and and edge $uv$ to $T$
\State Add the edge $uv$ to every path with an endpoint $u$. 
\EndFor
\Statex
\State Construct the graph $G$ whose vertices are the leaves $\cl_r$ of $T$ 
\For{every path $P$ of $\TT$}
\State Let  $\ell$ and $\ell'$ be the endpoints of $P$.
\State Add the edge $\ell \ell'$ to $G$.
\EndFor
\State Compute $b_v$ for every vertex $v$ of $T$ using \eqref{eqn:degreeBounds}.
\State $H \gets $ an optimum solution of the hierarchical $b$-matching instance $(G,\vect{b})$.
\State $U_H \gets $ the paths of $\TT$ corresponding to the edges of $H$.
\State \Return $U \cup U_E$.
\end{algorithmic}
\end{algorithm}

%\section{Maximum k-colorable set}\label{sec:maxk}
%\input{MaxkColorableSet}

\section{Conclusion, Simulations and Open Problems}\label{sec:conc}
In this paper, we presented hardness results and an approximation algorithm for the problem of packing a maximum number of subtrees of a tree with heterogeneous capacity constraints.
This problem generalizes other problems in the literature.
Our analysis implies some of existing results in the literature regarding the maximum $k$-colorable subgraph problem in chordal graphs
and the maximum multi-commodity flow problem in trees.
Furthermore, it reveals an optimal subcase unknown so far.

We also presented an optimal algorithm for a special case in which all the paths are either directed or cross the root,
thus extending a known optimal algorithm for paths in star networks,
It is evident that the case where the capacities of all internal vertices and internal edges of the tree are unbounded,
or at least sufficient capacity is provisioned in the internal vertices and edges to handle all possible communication requests between the leaves of the tree,
is equivalent to a star in which the capacity of the center of the star is unbounded.
Surprisingly, such configurations can be found in data center architectures in practice (e.g., \cite{Roy2015Traffic}).

We conducted simulations to observe the approximation ratio of the greedy algorithm {\alg} in practice for paths on a tree with high capacities. 
We compared the performance of the greedy algorithm to the optimum which is computed by using algorithm \ref{alg:VPS}.
The approximation ratio was at most 1.6 with an average of 1.25 on random instances, and at most 1.005 in data center traffic as reported by \cite{Roy2015Traffic}. 
Plots of these results can be found in Figure \ref{fig:Random} and Figure \ref{fig:RealData} in the Appendix.

We now mention some open problems immediately relevant to our work. 
1) To the best of our knowledge, it is unknown whether the maximum $k$-colorable set problem in chordal graphs admits a $\ptas$. 
This is an interesting open question.
2) We considered the heterogeneous capacities case and showed that it is much harder to approximate than the homogeneous case.
It makes sense to analyze the case of bounded heterogeneity, i.e. when the ratio of the largest capacity to the smallest is bounded.
3) The most important extension seems to be the weighted case in which there is a profit associated with every subtree and one has to find
 a subset of trees that maximizes the total profit under given capacity constraints.

%%
%% Bibliography
%%

%% Please use bibtex, 

\newpage
\bibliography{GraphTheory,Optical,Mordo,Matching,Approximation,References}

\newpage
\appendix
\section{Simulation Results}

\begin{figure*}[ht]
	\centering
	\includegraphics[width=0.8\textwidth]{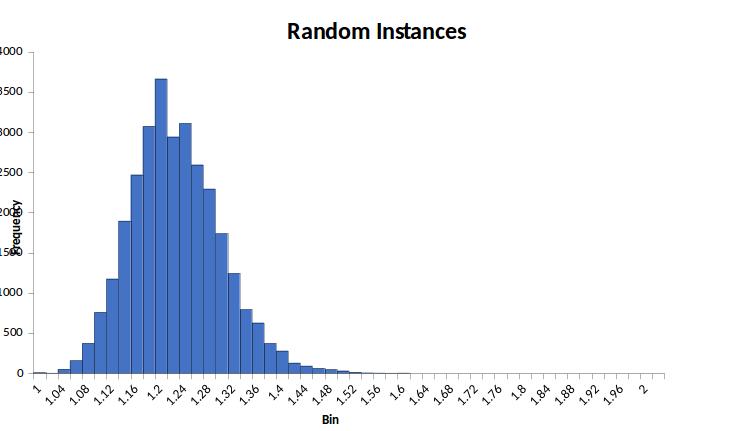}
	\caption{Approximation ratio of the greedy algorithm on random $\vpt$ graphs. The histogram is for 30.000 instances on random trees of 50 to 150 nodes, and number of paths from twice to four times the size of the tree.}
	\label{fig:Random}
\end{figure*}

\begin{figure*}[ht]

	\begin{subfigure}[t]{\textwidth}
	    \centering
		\includegraphics[width=0.8\textwidth]{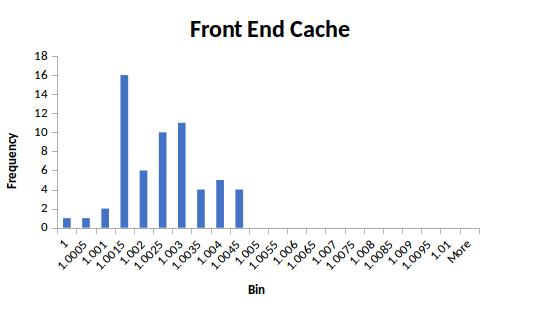}
		\label{fig:FrontEndCache}
	\end{subfigure}

	\begin{subfigure}[t]{\textwidth}
	    \centering
		\includegraphics[width=0.8\textwidth]{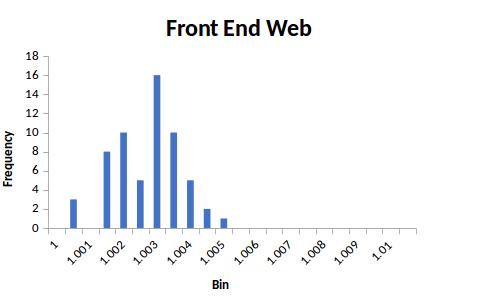}
		\label{fig:FrontEndWeb}
	\end{subfigure}

	\begin{subfigure}[t]{\textwidth}
	    \centering
		\includegraphics[width=0.8\textwidth]{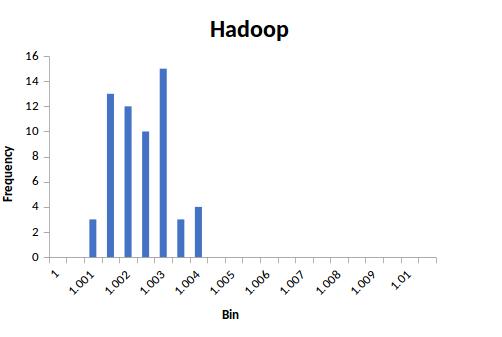}
		\label{fig:Hadoop}
	\end{subfigure}
	
	\caption{Approximation ratio of the greedy algorithm for 180 instances on data center traffic according to the distributions in \cite{Roy2015Traffic}. The individual histograms show the ratio for different traffic distributions resulting from three different applications.}
	\label{fig:RealData}

\end{figure*}
 
\end{document}